\documentclass[conference,letterpaper]{IEEEtran}

\IEEEoverridecommandlockouts

\usepackage{cite}
\usepackage{amsmath,amssymb,amsfonts}
\usepackage{graphicx}
\usepackage{textcomp}
\usepackage{xcolor}
\usepackage{algorithm}
\usepackage{algpseudocode}
\usepackage{amsthm}
\usepackage{bbm}
\usepackage{caption}
\usepackage{comment}
\usepackage{diagbox}
\usepackage{etoolbox}
\usepackage{float}
\usepackage{paralist}
\usepackage{import}
\usepackage{subcaption}
\usepackage{tabularx}
\usepackage{todonotes}
\usepackage{stmaryrd}
\usepackage{balance}
\usepackage{url}
\usepackage[utf8]{inputenc}

\newtheorem{theorem}{Theorem}%
\newtheorem{lemma}[theorem]{Lemma}%
\newtheorem{definition}[theorem]{Definition}%

\def\BibTeX{{\rm B\kern-.05em{\sc i\kern-.025em b}\kern-.08em
    T\kern-.1667em\lower.7ex\hbox{E}\kern-.125emX}}
\begin{document}

\title{Secure Integrated Sensing and Communication\\ against Communication and Sensing Eavesdropping
}

\author{
\IEEEauthorblockN{Sidong Guo and Matthieu R. Bloch}
\IEEEauthorblockA{School of Electrical and Computer Engineering 
Georgia Institute of Technology 
Atlanta, GA 30332\\
Email: sguo93@gatech.edu, matthieu@gatech.edu}
}

\maketitle

\makeatletter
\def\th@plain{%
  \thm@notefont{}
  \normalfont
}
\makeatother
\newtheoremstyle{myremark}{}{}{\itshape}{}{\itshape}{:}{ }{}
\theoremstyle{myremark}

\setlength{\abovedisplayskip}{4pt}
\setlength{\belowdisplayskip}{4pt}
\setlength{\abovedisplayshortskip}{2pt}
\setlength{\belowdisplayshortskip}{2pt}

\begin{abstract}
Sensing privacy and communication confidentiality play fundamentally different but interconnected roles in adversarial wireless environments.
Capturing this interplay within a single physical-layer framework is particularly challenging in integrated sensing and communication (ISAC) systems, where the same waveform simultaneously serves dual purposes. We study a secure ISAC system in which a monostatic transmitter simultaneously sends a confidential message to a legitimate receiver and senses an environmental state, while a passive adversary attempts both message decoding and state estimation. We partially characterize the fundamental trade-offs among three performance measures: the transmitter’s secrecy rate, its detection exponent, and the adversary’s detection exponent. Beyond the joint input distribution that governs overall performance, the trade-offs are further shaped by the transmitter’s ability to extract keys via feedback and hide both the \textit{content} and \textit{structure} of the codewords via wiretap and resolvability codes. We derive an achievable region, and illustrate the resulting design trade-offs through a numerical example.
\end{abstract}

\begin{IEEEkeywords}
Integrated sensing and communication (ISAC), physical-layer security, sensing security, detection exponent, resolvability
\end{IEEEkeywords}

\section{Introduction}\label{sec:introduction}
The integration of sensing and communication (ISAC) envisioned for 6G~\cite{wei2022toward,zhang2021physical} introduces security challenges absent in traditional separated architectures. Because a single waveform simultaneously conveys messages and encodes information about the environment through its structure, an adversary may exploit communication signals for message eavesdropping and sensing inference. 
Secure ISAC has been studied from a communication-theoretic perspective, with a primary focus on waveform optimization to maximize secrecy rate under various detection constraints \cite{ren2023robust,chu2022joint,he2024joint}. In parallel, a number of information-theoretic studies have emerged. Notably, \cite{gunlu2022secure} extends the i.i.d. ISAC model of \cite{ahmadipour_information-theoretic_2022,ahmadipour2023information} to incorporate a passive adversary whose state is to be estimated. The resulting secrecy–distortion region is characterized by leveraging results from wiretap channels \cite{bloch2021overview,bassi2018wiretap} and secret-key generation \cite{hayashi2016secret}. Nevertheless, despite repeated mentions in the literature for modeling and theoretical results that address \textit{joint} communication and sensing eavesdropping, such work remains scarce \cite{su2025integrating,aman2025integrating,ren2024secure}. 

The present work develops an information-theoretic framework for secure ISAC in the presence of an adversary with \textit{dual} eavesdropping objectives. Unlike the i.i.d. models commonly studied in prior work \cite{ahmadipour_information-theoretic_2022,ahmadipour2023information,gunlu2022secure}, we consider a \textit{fixed-state} model that enables both the transmitter and adversary to learn via sequential decision-making \cite{chang_sequential_2023,chang2023rate}. In realistic settings, the legitimate receiver’s and eavesdropper’s channels are often correlated, implying that the adversary can generally exploit transmitter adaptation for state inference. Consequently, we characterize the transmitter’s ability to \textit{slow down} the adversary’s state detection, quantified by the asymptotic detection exponent when the state remains constant \cite{chang2021evasive}.

The tradeoff between communication and sensing often amounts to characterizing an input distribution that simultaneously governs reliability and detection performance \cite{ahmadipour_information-theoretic_2022,chang2023rate}. The addition of security and the ability to extract secret keys accentuates this tension. On one hand, slowing the adversary’s ability to sense requires the transmitter to induce unfavorable output statistics at the adversary \cite{han2002approximation}. On the other hand, maintaining confidentiality in adverse channel conditions also relies on the judicious use of keys. To characterize these intertwined objectives, we draw on several established tools from literature: sequential hypothesis testing \cite{chang2021evasive,chernoff1992sequential,nitinawarat2013controlled,nitinawarat2015controlled}, compound wiretap channels with feedback \cite{bjelakovic2013secrecy,ahlswede2006transmission,ahlswede2002common,bassi2018wiretap,liang2009compound,ardestanizadeh2009wiretap}, and channel resolvability \cite{han2002approximation,bloch2016covert,bloch2013strong}. 

Our key conceptual contribution is to highlight a fundamental distinction between information-theoretic secrecy and sensing security in wireless systems: achieving the latter requires obscuring not only the content of the transmitted codewords but also their structural dependence on the underlying state.  Within this framework, our specific contributions are as follows:
\begin{inparaenum}[1)]
\item We show that the sensing-security-optimal operating point \(P_{\textnormal{SO}}\) is asymptotically achievable by appropriately designing state-dependent codeword types via resolvability coding; 
\item We demonstrate that a higher secrecy rate can be achieved in the resolvability regime \(P_{\textnormal{SC}}\) by selecting codeword types that damage Eve more than legitimate parties, at the cost of increased sensing leakage; this performance is ultimately limited by the soft-covering rate~\cite{yagli2019exact};
\item We derive the achievable secrecy rate in the non-resolvability regime \(P_{\textnormal{CO}}\).
\end{inparaenum}

\section{Notation}\label{sec:Notation}
Let \(X\) be a random variable taking values in a finite alphabet \(\mathcal{X}\). For \(n\in\mathbb{N}\), we write \(X^n\triangleq (X_1,\ldots,X_n)\in\mathcal{X}^n\), and for integers \(1\le i\le j\le n\), we denote the subsequence by \(X_{i:j}\triangleq (X_i,\ldots,X_j)\). For a deterministic sequence \(\mathbf{x}\in\mathcal{X}^n\), its type is \(P_{\mathbf{x}}(x)\triangleq \frac{1}{n}\sum_{t=1}^n \mathbbm{1}\{x_t=x\}\), and \(\mathcal{T}_{P}^n\) denotes the type class associated with a probability mass function (pmf) \(P\)~\cite{thomas2006elements}.  For conditional pmfs \(W_1(\cdot|x)\) and \(W_2(\cdot|x)\) on \(\mathcal{Y}\) given \(x\in\mathcal{X}\), and an input distribution \(P\in\mathcal{P}(\mathcal{X})\), define the conditional KL divergence and conditional Chernoff information by \(\mathbb{D}(W_1\Vert W_2\mid P)\triangleq \sum_{x\in\mathcal{X}} P(x)\,\mathbb{D}(W_1(\cdot|x)\Vert W_2(\cdot|x))\) and \(\mathbb{C}(W_1\Vert W_2\mid P)\triangleq \max_{\lambda\in[0,1]}\big(-\sum_{x\in\mathcal{X}} P(x)\log \sum_{y\in\mathcal{Y}} W_1(y|x)^{\lambda}W_2(y|x)^{1-\lambda}\big)\). The total variation distance between pmfs \(P_X\) and \(Q_X\) on \(\mathcal{X}\) is \(\mathbb{V}(P_X,Q_X)\triangleq \frac{1}{2}\sum_{x\in\mathcal{X}} |P_X(x)-Q_X(x)|\). Given an input distribution \(P_X\) and a channel \(W_{Y|X}\), we write \(P_X\circ W_{Y|X}\) for the induced output distribution \((P_X\circ W_{Y|X})(y)\triangleq \sum_{x\in\mathcal{X}} P_X(x)W_{Y|X}(y|x)\), and denote the mutual information under the joint law \(P_XW_{Y|X}\) by \(\mathbb{I}(P_X;W_{Y|X})\). For a joint channel \(W_{Y,Z|X}\), the conditional entropy of \(Y|Z\) under \(P_XW_{Y,Z|X}\) is denoted by \(\mathbb{H}_{P_X\circ W_{Y,Z|X}}(Y|Z)\). Finally, for a codebook \(\mathcal{C}_M^n=\{x^n(m)\}_{m=1}^M\) of size \(M\) used with a uniform message prior, the induced (mixture) output distribution over a channel \(W_{Y|X}\) is \(P_{Y^n\mid \mathcal{C}_M^n}(y^n)\triangleq \frac{1}{M}\sum_{m=1}^M W_{Y^n|X^n}(y^n\mid x^n(m))\).

\section{System Model and Main Results}\label{sec:systemModel}

Consider the model illustrated in Fig.~\ref{fig:1}, in which the transmitter (Tx) communicates over a state-dependent discrete memoryless channel (DMC) while simultaneously performing a multiple-hypothesis test over the state space \(s \in \mathcal{S} \triangleq \{0, \dots, \Theta-1\}\). The legitimate receiver (Rx) is tasked with decoding the message while an eavesdropper (Eve) attempts both state estimation and message decoding. A sequential message set is defined as \(\mathcal{M} = \cup_{t=1}^\infty \mathbb{F}^t_2, \mathbb{F}_2 \triangleq \{0,1\}\) \cite{chang_sequential_2023}, which is to be kept secure from Eve. A sequential detection framework creates an asymmetry between Tx, who can adaptively decide when to stop transmission, and a passive Eve, thereby granting the legitimate parties an inherent advantage,  which often exists in practice. The channel is specified by the transition probability \(W_{Y_1Y_2Z|XS}\), where \(Y_1\) denotes the Rx's observation, \(Y_2\) Eve's observation, and \(Z\) the feedback signal at Tx. The feedback is assumed to be noiseless, secure, and one-tap delayed, so that \(Z_t = Y_{1,t}\), and the state \(S\) fixed throughout the transmission. Note that the model integrates a \emph{monostatic}, fixed-state ISAC setting with a dual-objective adversary operating in a \emph{bistatic} configuration. 
 We make the following modeling assumptions: 
 \begin{enumerate}
     \item The input, output and state spaces \(\mathcal{X}, \mathcal{Y}_1, \mathcal{Y}_2, \mathcal{Z}, \mathcal{S}\) are all finite.
     \item \(\forall s \neq s' \in \mathcal{S}\) and \(\forall x\in \mathcal{X}\), \(\forall y_1,y_2\in \mathcal{Y}_1,\mathcal{Y}_2\), \(0< \mathbb{D}(W_{\cdot|X,s}(\cdot|x)||W_{\cdot|X,s'}(\cdot|x)) < \infty\).
     \item The prior distribution over all hypotheses is uniform, i.e., \(\pi_s = \frac{1}{\Theta} \forall s \in \mathcal{S}\).
     \item Both Tx, Rx and Eve have knowledge of the set of kernels \(\{W_{Y_1|X,s},W_{Y_2|X,s}\}_{s\in \mathcal{S}}\). 
 \end{enumerate}
 Assumptions 1) and 2) are technical and necessary for finite-state multi-hypothesis testing to ensure all channels are distinguishable and that no observation is noise-free. Assumption 3) ensures the optimality of ML detector. Assumption 4) differs from earlier communication and information theoretic works in which Eve only has access to its own channels \cite{ren2023robust,chu2022joint,he2024joint}\cite{chang2021evasive}. This stronger requirement is motivated by the ISAC setting, in which Tx and Eve's channels are typically correlated through the same underlying environmental state. As a result, Eve may extract state information not only from its direct observations but also indirectly from Tx’s adaptive behavior, particularly when its channel is stronger. 

\begin{figure}[htb!]
    \centering
    \includegraphics[width=8.5cm]{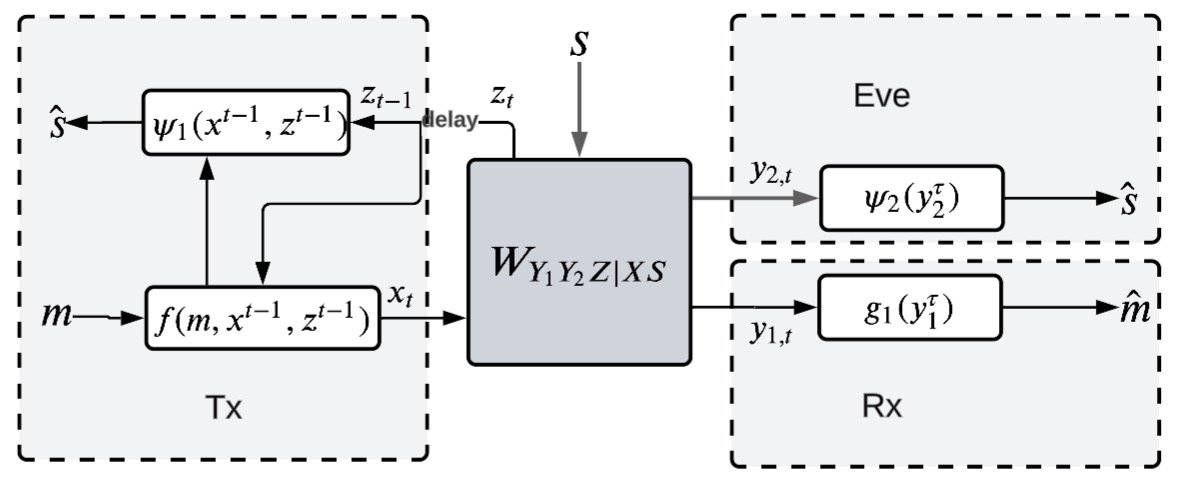}
    \caption{System Model}
    \label{fig:1}
\end{figure}

We now describe the operation of the system in detail. 
For each round \(t \geq 1\), Tx employs an encoding function to produce the next input to the channel as \(f_t:\mathcal{M}\times \mathcal{X}^{t-1}\times \mathcal{Y}_{1}^{t-1} \mapsto \mathcal{X}\). Tx concludes the communication at round \(\tau+1\), where \(\tau\) is a random stopping time determined by a stopping rule \(h_t:\mathcal{X}^{t}\times \mathcal{Y}_1^{t}\mapsto \mathbb{F}_2\). Let \(\mathbb{P}_s(X^t,Y_1^t)\) be the joint probability of \(X^t,Y_1^t\) under state \(s\), the stopping rule is based on the log-likelihood ratio (LLR)
\begin{equation} \label{stoppingrulegeneral}
    \tau =  \min_{s \in \mathcal{S}} \inf\bigg\{t:  \log \frac{\mathbb{P}_s(X^t,Y_1^t)}{\max_{s'\neq s}\mathbb{P}_{s'}(X^t,Y_1^t)} \geq  c_s \bigg\},
\end{equation}
for \(c_s>0\) depending on \(s\).
For each round \(t\in [1,\tau]\), given lossless feedback, the estimated state at Tx is produced by the maximum likelihood estimator (MLE)
\begin{align}\label{MLE}
    \psi_1(x^{t-1}, y_1^{t-1})   
     = \underset{s\in \mathcal{S}}{\arg \max} \, \prod_{k=1}^{t-1} W_{Y_1|X,s}(y_{1,k}|x_k),
\end{align}
where at stopping time \(\tau\) (round \(\tau+1\)), Tx declares an estimated state \(\hat{S} = \psi_1(X^\tau, Y_1^{\tau})\).
On the Rx side, the decoding function for message \(M\) is defined as \(g_1:\cup_{t\geq1}\mathcal{Y}_1^t \mapsto \mathcal{M}\).
 
Eve seeks to simultaneously infer the transmitted message and estimate the state in competition with Tx. Although it cannot influence the transmission policy, Eve may exploit Tx's adaptation by forming its own state estimate at the end of transmission, \(\psi_2:\cup_{t\geq1}\mathcal{Y}_2^{t} \mapsto \mathcal{S}\). It is also possible for Eve to employ the joint estimator \(\arg \max_{m,s} \mathbb{P}(y_2^t |M=m,S=s)\) \cite{chang2023rate}. However, since strong secrecy will be required as part of the achievable definition, Eve's optimal detection is instead given by 
 \begin{equation}
    \psi_2(y_2^t) = \arg \max_{s\in \mathcal{S}} \sum_m \mathbb{P}(y_2^t|M=m,S=s). 
 \end{equation}
By Assumption 3), Eve's decision at stopping time \(\tau\) is the ML estimation 
 over its marginalized mixture posterior \(\mathbb{P}_s(y_2^\tau)\) 
 \begin{equation}\label{eve posterior} 
\begin{aligned}
\frac{1}{2^{M_\tau}}
   \sum_{m=1}^{2^{M_\tau}}
   \Bigg[
      P_{X,s}\!\big(x_1(m)\big)     
      \prod_{k=1}^{\tau-1}
      P_{X,s}\!\big(x_{k+1}(m)\mid x^{k}(m)\big)
 \\ \times
      \prod_{k=1}^{\tau-1}
      W_{Y_2|X,s}\!\big(y_{2,k} \mid x_k(m)\big)
   \Bigg].
\end{aligned}
\end{equation}

For all \(s\in \mathcal{S}\), the reliability is measured by the asymptotic error probability as
 \begin{equation}
     P_c = \underset{s \in \mathcal{S}}{\mathrm{max}} \, \underset{m \in \mathcal{M}}{\mathrm{max}} \,\, \mathbb{P}(g_1(Y_1^\tau) \neq M[1;M_\tau] |M = m,S = s),\nonumber
 \end{equation}
 where \(M_t\) is the number of message bits transmitted by time \(t\) and \(M[1;M_\tau]\) represents the first \(M_\tau\) bits of the message \(M= m\) transmitted by stopping time.
 On the other hand, secrecy is measured in terms of the leakage of message \(M\) to Eve up to the stopping time \(\tau\) as \(\underset{s \in \mathcal{S}}{\mathrm{max}}\,\mathbb{I}(M;Y^\tau_2)\). The detection performance of both Rx and Eve are measured through the detection error probability defined as 
\begin{align}
    &P_{d_1} =  \underset{s \in \mathcal{S}}{\mathrm{max}} \, \underset{m \in \mathcal{M}}{\mathrm{max}} \,\, \mathbb{P}(\psi_1(X^\tau, Y_1^{\tau}) \neq S |M = m,S = s), \nonumber\\
    &P_{d_2} =  \underset{s \in \mathcal{S}}{\mathrm{max}} \,\, \mathbb{P}(\psi_2(Y_2^\tau) \neq S | S = s).\nonumber
\end{align}
The rate and the detection-error exponents are then defined as 
\begin{equation} \label{definitions}
    R^n \triangleq \frac{M_\tau}{n}, E^n_{d_1} = -\frac{1}{n}\log P_{d_1}, E^n_{d_2} = -\frac{1}{n}\log P_{d_2},
\end{equation}
where \(n\) is a constraint on the stopping time.
\begin{definition}
\textit{A} Let \(\delta_1,\delta_2,\delta_3,\delta_4,\delta_5,\delta_6 >0\), \((R, E_1, E_2)\) \textit{tuple is achievable if} \(\exists n\geq n(\delta_1,\delta_2,\delta_3,\delta_4,\delta_5,\delta_6)\) \textit{such that ISAC policy} \((\{h_t\}_{t\geq 1}, \{f_t\}_{t\geq 1}, g_1, \psi_1)\) \textit{has vanishing communication error probability} \( P_c \leq \delta_1\), \textit{information leakage} \(\underset{s \in \mathcal{S}}{\mathrm{max}}\,\mathbb{I}(M;Y_2^\tau) \leq \delta_2\) \textit{and probabilistic stopping time} \(\underset{s \in \mathcal{S}}{\mathrm{max}} \,\underset{w \in \mathcal{M}}{\mathrm{max}} \,\, \mathbb{P}(\tau >n) \leq \delta_3\) \textit{with} 
\begin{align}
    &\underset{s \in \mathcal{S}}{\mathrm{min}} \,\underset{w \in \mathcal{M}}{\mathrm{min}} \,\, \mathbb{P}(R^n \geq R) \geq 1-\delta_4, \\
    & E_{d_1}^n \geq E_1 -\delta_5 ,\\
    & E_{d_2}^n \leq E_2 +\delta_6.
\end{align}
\end{definition}

\begin{theorem} \label{achievability theorem}
\textit{The following region is achievable}:
\begin{equation}
    \bigcup_{\substack{
    \rho \in [0,1],\\
    \{P_{X,s}\}_{s \in \mathcal{S}} \in (\mathcal{P}_\mathcal{X})^{|\mathcal{S}|}
}}
\bigcap_{s\in \mathcal{S}} \{
            (R,E_1,E_2) \in \mathbb{R}^3_{+}\},
\end{equation}
\begin{align}
    &R \leq \min (\rho R_{\text{key}}(s)  \nonumber\\& \quad\quad\quad+ (1-\rho)[R_1(s) -R_2(s)+ R_{\text{key}}(s)]^+, R_1(s)), \nonumber\\
    &E_1  \leq \min_{s'\neq s} \mathbb{D}(W_{Y_1|X,s}||W_{Y_1|X,s'}|P_{X,s}), \nonumber\\
    &E_2  \geq \min_{s'\neq s} \mathbbm{1}_{\{P_{X,s},P_{X,s'} \in \mathcal{P}_{\textnormal{res}}\}} \{\rho\mathbb{C}(W_{Y_2|X,s}||W_{Y_2|X,s'}|P_{X,s}) \nonumber\\
            & \,\,\,\,\,\,\,\,\,\,\,\,\,\,\,\,\, +(1-\rho) \mathbb{C}(P_{X,s}\circ W_{Y_2|X,s}||P_{X,s'}\circ W_{Y_2|X,s'})\} \nonumber\\
            &\,\,\,\,\,\,\,\,\,\,\,\,\,\,\,\,\,+\mathbbm{1}_{\{P_{X,s},P_{X,s'} \notin \mathcal{P}_{\textnormal{res}}\}}\{\mathbb{C}(W_{Y_2|X,s}||W_{Y_2|X,s'}|P_{X,s})\} \nonumber,
\end{align}

\textit{where} \(R_{\text{key}}(s) = \mathbb{H}_{P_{X,s}\circ W_{Y_1,Y_2|X,s}}(Y_1|X,Y_2)\), \(R_1(s) = \mathbb{I}(P_{X,s};W_{Y_1|X,s})\), \(R_2(s) = \mathbb{I}(P_{X,s};W_{Y_2|X,s})\). \(\mathcal{P}_{\textnormal{res}}\) \textit{is the resolvability region defined as}

\begin{equation}
\left\{
\substack{P_{X,s}\\P_{X,s'}} \in \mathcal{P}_{\mathcal X}:\;
\begin{aligned}
& R_1(s) - R_2(s) + R_{\text{key}}(s) \ge 0, \\[2mm]
& \mathbb{C}\!\left(P_{X,s}\!\circ W_{Y_2|X,s}\,\big\|\,P_{X,s'}\!\circ W_{Y_2|X,s'}\right) \\[2mm]
& < \min_{s''\in\{s,s'\}} E_{\textnormal{SC}}(P_{X,s''}, W_{Y_2|X,s''}),\nonumber
\end{aligned}
\right\},
\end{equation}
\begin{align}
    E_{\textnormal{SC}}(P_X,W_{Y|X}) = \min_{Q_{Y|X} \in \mathcal{P}_{\mathcal{Y|X}}} \bigg\{\mathbb{D}(P_X Q_{Y|X} ||P_X W_{Y|X})\nonumber\\+\frac{1}{2}[R - \mathbb{D}(P_X Q_{Y|X}||P_X Q_Y)]^+\bigg\}\nonumber
\end{align}
\textit{is the soft covering exponent} \cite{yagli2019exact}.

\end{theorem}

Theorem~\ref{achievability theorem} follows from an adaptive coding strategy that adjusts its operation depending on whether soft covering \cite{yagli2019exact} provides sufficient randomness to conceal the structure of the mixture codebook at the eavesdropper. When the soft-covering rate exceeds the i.i.d.\ Chernoff detection exponent and Tx enjoys a stronger channel than Eve, Tx employs a key-assisted wiretap code, thereby enhancing both communication secrecy and sensing privacy.

In contrast, when Eve has a moderate channel advantage, Tx may choose between two options:
(i) a resolvability-assisted scheme that uses a portion of the key to induce approximate product distributions at Eve and the remaining key for one-time-pad (OTP) encryption, thereby securing both communication and sensing; or
(ii) a pure communication-security scheme that relies solely on OTP, which achieves higher secrecy rate at the cost of reduced sensing privacy.
The achievable region in this case is obtained through a simple time-sharing argument between these two policies (captured by parameter \(\rho\)).

Finally, when Eve’s channel advantage is so pronounced Tx cannot mask the structure of the codebook. Analyzing this regime is challenging, since no closed-form expression exists for Chernoff information between mixtures. In this case, Tx allocates all keys to OTP-based communication security and reverts to an open-loop sensing strategy to minimize state information leakage.

\section{Numerical Example}
For illustration purposes, we consider the channel specified in Table~\ref{tab:channel_parameters}. In this example, all channels 
\(\{W_{Y_1|X,s}\}_{s\in\{0,1\}}\) and \(\{W_{Y_2|X,s}\}_{s\in\{0,1\}}\) are binary symmetric channels (BSCs). 
For BSCs, it is known that there is no tradeoff between Tx’s sensing exponent and its 
communication rate~\cite{chang2023rate}; accordingly, we plot the \((R,E_1)\) and \((R,E_2)\) regions 
in Fig.~\ref{fig:2}. In this setting, we operate in the regime
\(
(R_2(s) - R_{\text{key}}(s))^{+} < R_1(s) < R_2(s), \forall\, s.
\)
The sensing-security–optimal point \(P_{\textnormal{SO}}\) (region I) is achieved through a confusion strategy in which Tx selects the input distributions \(\{P_{X,s}\}_{s\in \mathcal{S}}\) to satisfy
\(
\mathbb{C}\!\left(P_{X,s}\circ W_{Y_2|X,s}\,\big\|\,P_{X,s'}\circ W_{Y_2|X,s'}\right) = 0, \forall s,s'\in \mathcal{S}.
\)
For BSCs, choosing \(P_{X,s}(0)= \tfrac{1}{2} \forall s\) holds universally. Region II is obtained with a resolvability code that strengthens sensing security at the cost of a smaller maximum secrecy rate. 
The transition from \(P_{\textnormal{SO}}\) to \(P_{\textnormal{SC}}\) occurs when Tx deviates from the sensing-security-optimal tuple to one that maximizes the secrecy rate. In this example, Eve’s detection 
exponent is smaller than the soft-covering exponent, so the latter does not impose a performance 
constraint.
Beyond the resolvability threshold, achieving the communication-security–optimal point \(P_{\textnormal{CO}}\) (region III) requires 
Tx to allocate all secret keys to OTP encryption and operate in an open-loop sensing mode, at the cost of a larger sensing leakage.
The segment connecting \(P_{\textnormal{SC}}\) and \(P_{\textnormal{CO}}\) is obtained by time-sharing between these two operating 
points. 

\begin{table}[htb!]
\caption{Table for BSC $W_{Y_2|X,S}(y_2|x,s),W_{Y_1|X,S}(y_1|x,s)$ and $S\in\{0,1\}$}
    \centering
    \begin{tabular}{|c |c |c |}
    \hline
         \diagbox{S}{X} & $W_{Y_1|X,S}(y_1=1|x=0,s)$ & $W_{Y_2|X,S}(y_2=1|x=0,s)$ \\
         \hline
        $0$ &\(0.1\) & \(0.06\) \\
        \hline
         $1$ &\(0.15\) & \(0.03\) \\
    \hline
    \end{tabular}
    \label{tab:channel_parameters}
\end{table}

\captionsetup{font={small}}
\begin{figure}
    \centering
    \includegraphics[width=9cm]{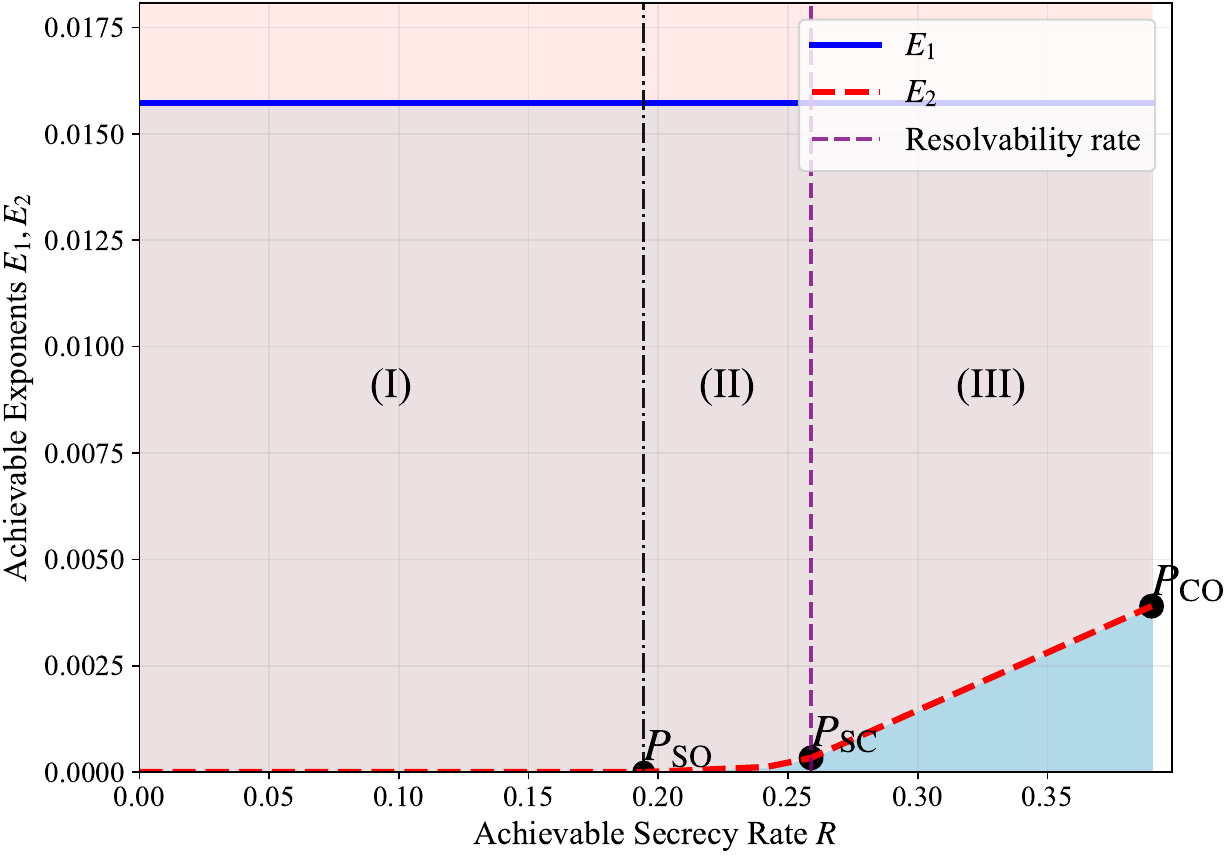}
    \caption{Illustration of Secrecy Privacy Tradeoffs in secure ISAC }
    \label{fig:2}
\end{figure}

\section{Proof of Theorem \ref{achievability theorem}}
\subsection{Policy Description}\label{Policy Description}
Tx employs a block-based feedback-adaptation scheme.
For \(n \in \mathbb{N}^*\), define \(N = \sqrt{n}\) as the block length. At the beginning of the protocol, Tx and Rx first agree on a public set of types \(\{P_{X,s}\}_{s\in \mathcal{S}}\) for each state. Let \(\hat{S}_{(k-1)N} = \psi_1(X^{(k-1)N},Y^{(k-1)N}_1)\) be the estimated state at the beginning of \(k\)-th block, then \(P_{X,\hat{S}_{(k-1)N}}\) is the codeword type used in \(k\)-th block. Exceptionally, for \(k \leq n^{\beta}, \beta \in (0,\frac{1}{2})\), we use a universal code with type \(\tilde{P}_X\).

For each block, Tx employs a different scheme based on estimated channels. i) When \(R_1(\hat{S}_{(k-1)N}) > R_2(\hat{S}_{(k-1)N})\), Tx uses wiretap coding with feedback \cite{ahlswede2006transmission}\cite[Chapter 3.6]{bloch2011physical}. ii) When \((R_2(\hat{S}_{(k-1)N})-R_{\text{key}}(\hat{S}_{(k-1)N}))^+<R_1(\hat{S}_{(k-1)N}) < R_2(\hat{S}_{(k-1)N})\), Tx devotes part of the key for resolvability coding, and uses the remaining key for secure communication, using OTP. iii) When \((R_2(\hat{S}_{(k-1)N})-R_{\text{key}}(\hat{S}_{(k-1)N}))^+ >R_1(\hat{S}_{(k-1)N})\), Tx allocates all key bits for OTP. 

\subsubsection{Encoding}
For \(k  \leq n^{\beta}\), Tx uses a universal code for sensing purpose only. This serves the dual purpose of obtaining an initial estimation of the channels as well as obtaining shared keys for subsequent blocks. Tx draws channel inputs from the distribution \(\tilde{P}_X = \arg \max_{P_X}  \min_{s, s'\neq s} \mathbb{D}(W_{Y_1|X,s}||W_{Y_1|X,s'}|P_{X}) \).

Encoding for \(k  > n^{\beta}\) is split into two portions, where a fraction \((1-\Delta)N\) of block length is used for transmitting message and the remaining \(\Delta N\) for communicating code structure as described later, for \(0<\Delta < 1\). In the sequel, we focus on (ii); policy and analysis of  (i) and (iii) are based on combining results on wiretap coding with feedback \cite{ahlswede2006transmission} and OTP-based secrecy \cite{bloch2011physical}. By \cite[Lemma 1]{ahlswede2006transmission}, for \(n\) sufficiently large from key distillation of the previous block, Tx receives secret key \(r_k\) uniformly distributed in \(\llbracket 1, 2^{N (R_{\text{key}}(\hat{S}_{(k-1)N})-\delta(\epsilon))}\rrbracket\), split into two parts as \((r_{1,k},r_{2,k})\). Construct a codebook with \(\lceil 2^{(1-\Delta)N (R_2(\hat{S}_{(k-1)N})-\delta(\epsilon))} \rceil\) codewords independently drawn from the distribution \(P_{X,\hat{S}_{(k-1)N}}\), relabeled as 
\begin{align} \label{index2}
    &(m_k,n_k,r_{1,k}) \nonumber\\& \in  \llbracket 1, 2^{(1-\Delta) N (R_1(\hat{S}_{(k-1)N})+R_{\text{key}}(\hat{S}_{(k-1)N})-R_2(\hat{S}_{(k-1)N})-\delta(\epsilon))^+}\rrbracket \nonumber \\& \times \llbracket 1, 2^{(1-\Delta) N(R_2(\hat{S}_{(k-1)N})-R_{\text{key}}(\hat{S}_{(k-1)N}))^+}\rrbracket \nonumber\\  &\times \llbracket 1, 2^{(1-\Delta) N(R_2(\hat{S}_{(k-1)N})-R_1(\hat{S}_{(k-1)N}))^+}\rrbracket .
\end{align}
To transmit message in \(k\)-th block, Tx sends \(\mathbf{x}^{(1-\Delta)N}(m_k\oplus r_{2,k}, n_k,r_{1,k})\). Tx uses the remaining block length \(\Delta N\) to transmit information about \(P_{X,\hat{S}_{(k-1)N}}\) with the same method, where codewords are drawn from \(\tilde{P}_X\). This has no rate impact since \(\lim_{n\rightarrow \infty} \frac{\log |\mathcal{S}|}{\Delta N}=0\). 

\subsubsection{Stopping rule and Decoding}
The communication concludes whenever (\ref{stoppingrulegeneral}) is satisfied. For \(t \equiv 1 (\text{mod}\,N)\), Tx chooses 
\begin{equation} \label{stopping rule}
    c_s =  n \left(\mathbb{D}(W_{Y_1|X,s}||W_{Y_1|X,s'}|P_{X,s}) -\epsilon\right)
\end{equation}
for \(\epsilon>0\) in (\ref{stoppingrulegeneral}). In other words, Tx only determines whether the transmission stops at \(t \equiv 0 (\text{mod}\,N)\). 

After Rx obtains \(Y_1^\tau\), the decoder operates on a block basis. For each block \(k\), Rx first reconstructs the codeword type used, \(P_{X,\hat{s}_{(k-1)N}}\), using constant-composition decoding with \(\tilde{P}_X\) ~\cite{csiszar2011information} and shared OTP randomness. Once the code structure is recovered, the message \(m_k\) is decoded similarly.
\subsection{Policy Analysis}
\begin{lemma}\label{stopping time lemma}
    \textit{For any} \(s, s' \in \mathcal{S}\) \textit{there exists}  \(\psi_{s,s'},\phi_{s,s'} > 0\) \textit{such that} 
\begin{equation}
    \mathbb{P}_s(\tau >n) \leq (|\mathcal{S}|-1) \exp \left(- \min_{s'\neq s}\frac{n \psi_{s,s'}^2}{2\phi^2_{s,s'}}\right).
\end{equation}
\end{lemma}
\textit{Proof:} \cite[Appendix A]{guo2026secure}. \hfill $\blacksquare$

We first provide a result on the stopping time constraint. Lemma~\ref{stopping time lemma} states a high probability concentration on the stopping time around \(n\) as a result of the stopping rule. 

\subsubsection{Achievable exponent analysis for transmitter and eavesdropper}
Rx's detection error probability satisfies 
\begin{align} \label{main exponent}
    P_{d_1} \leq &\underset{s \in \mathcal{S}}{\text{max}}\,\underset{m \in \mathcal{M}}{\text{max}}\,(|\mathcal{S}|-1)\nonumber \\ &\times \exp(-n (\min_{s'\neq s}\,\mathbb{D}(W_{Y_1|X,s}||W_{Y_1|X,s'}|P_{X,s})-\epsilon)).
\end{align}
This comes from the fact that detection exponent result for sequential hypothesis testing holds under an input constraint \cite{nitinawarat2013controlled}, where a detailed proof can be found in \cite[Appendix B]{guo2026secure}. 

Next, Eve's detection error probability  is lower bounded as 
\begin{align}
    P_{d_2} = & \underset{s}{\max} \bigg[\sum_{t=1}^{n} \mathbb{P}_s(\psi_2(Y_2^\tau) \neq s |\tau = t) \mathbb{P}_s(\tau = t) \nonumber\\
    &+ \sum_{t=n+1}^\infty \mathbb{P}_s(\psi_2(Y_2^\tau) \neq s |\tau = t) \mathbb{P}_s(\tau = t) \bigg ] \nonumber\\
    \geq & \underset{s}{\max}\,\mathbb{P}_s(\tau \leq n)\mathbb{P}_s(\psi_2 (Y_2^{\tau}) \neq s |\tau = n) .
\end{align}
For \(n \rightarrow \infty \), \(\mathbb{P}_s(\tau \leq n) \geq 1-e^{-n\delta_3}\) by Lemma~\ref{stopping time lemma}. Moreover,
\begin{equation}
    \underset{s}{\max}\,\mathbb{P}_s(\psi_2(Y_2^n)\neq s) 
    \geq \max_{s,s':s'\neq s}\,(1-\delta_3)\mathbb{P}_s(f^{s,s'}_{ML}(Y_2^{n}) \neq s ),\nonumber
\end{equation}
where \(f^{s,s'}_{ML}\) is the binary MLE between \(s\) and \(s'\), we obtain
\begin{align} \label{achievableE2}
   -\frac{1}{n} \log P_{d_2} 
   \leq & \underset{s}{\min}\,\underset{s'\neq s}{\min}
   \left\{- \frac{1}{n} \log\mathbb{P}_s\left(f^{s,s'}_{ML}(Y_2^{n}) \neq s \right) + \delta_6\right\} .
\end{align}
We bound the piece-wise probability of detection error at Eve with the following lemma. 

\begin{lemma} \label{Eve exponent lemma}
Let \(s\in \mathcal{S}\) be the true hypothesis and \(\{P_{X,s}\}_{s \in \mathcal{S}}\) be the policy tuple, then
if \(R_1(s) +R_{\text{key}}(s)> R_2(s)\),
\begin{align}
     &- \frac{1}{n} \log\mathbb{P}_s\left(f^{s,s'}_{ML}(Y_2^{n}) \neq s \right)\nonumber\\ &\leq \mathbb{C}(P_{X,s}\circ W_{Y_2|X,s}||P_{X,s'}\circ W_{Y_2|X,s'})\nonumber \\ & -\frac{1}{n}\log\bigg(1-c_3 \frac{e^{-n\min (E_{\textnormal{SC}} (P_{X,s},W_{Y_2|X,s}),E_{\textnormal{SC}}(P_{X,s'},W_{Y_2|X,s'}))}}{e^{-n \mathbb{C}(P_{X,s}\circ W_{Y_2|X,s}||P_{X,s'}\circ W_{Y_2|X,s'})}}\bigg),
\end{align}
 for constant \(c_3 >0\). If \(R_1(s) +R_{\text{key}}(s)< R_2(s)\),
 \begin{equation}
     - \frac{1}{n} \log\mathbb{P}_s\left(f^{s,s'}_{ML}(Y_2^{n})\neq s \right) \leq \mathbb{C}(W_{Y_2|X,s}|| W_{Y_2|X,s'}|P_{X,s})
\end{equation}
\end{lemma}
\textit{Proof:} \cite[Appendix C]{guo2026secure}. \hfill $\blacksquare$

Lemma~\ref{Eve exponent lemma} states that in order for Tx to hide the structure of the code via channel resolvability, the minimum rate at which the mixture approximates the product distribution needs to be faster than the i.i.d Chernoff detection exponent, in which case Tx employs an adaptive strategy described in Section. \ref{Policy Description}. Otherwise, Tx cannot ensure sensing is secure and falls back on an open-loop strategy in order to limit the state information leakage via adaptive control sequence. 
\subsubsection{Secrecy and reliability analysis}
We first analyze the leakage chained across a random number of blocks.  

Let \(Y_{2,k} = Y_{2,(k-1)N+1:kN}\) and \(\tilde{M}_k = m[M_{(k-1)N};M_{kN}]\) denotes Eve's observation and the message transmitted in block \(k\). Define the binary variable \(A_k\) such that \(A_k = 1 \) when \(\psi_1(X^{(k-1)N},Y_1^{(k-1)N}) = s\) and \(A_k=0\) otherwise, where \(s\) is the true state. The information leakage \(\mathbb{I}(M;Y_2^\tau)\) is 
\begin{align} \label{secrecy}
    & \mathbb{I}(\tilde{M}_{n^{\beta}+1}\dots \tilde{M}_{\frac{\tau}{N}};Y_{2,n^{\beta}+1}\dots Y_{2,\frac{\tau}{N}}) \nonumber\\
    \leq & \sum_{k=n^{\beta}+1}^{\tau/N}\mathbb{I}(\tilde{M}_k;A_k|\tilde{M}_{1:k-1}) + \mathbb{I}(\tilde{M}_k;Y_2^\tau |\tilde{M}_{1:k-1},A_k)\nonumber \\
     \leq &\sum_{k=n^{\beta}+1}^{\tau/N} \mathbb{I}(\tilde{M}_k;Y_2^\tau|\tilde{M}_{1:k-1},A_k=1) + c_1\sqrt{n}e^{-c_2 k \sqrt{n}}
\end{align}
for constants \(c_1,c_2>0\). The last inequality in (\ref{secrecy}) comes from observing that \(A_k\) depends only on codeword type, which is independent of message \(\tilde{M}_k\) , upper bounding \(\mathbb{I}(\tilde{M}_k;Y_2^\tau|\tilde{M}_{1:k-1},A_k=0)\) by \(\log |\mathcal{M}_k|\) and  Lemma 5 in Appendix \cite{guo2026secure}. First term in (\ref{secrecy}) is the leakage within each block, which we can show is vanishing by standard joint typicality analysis on wiretap code \cite{thomas2006elements} and OTP \cite{bloch2011physical}. We relegate reliability analysis to \cite[Appendix D]{guo2026secure}. 

\subsubsection{Achievable secrecy rate analysis}
 For any \(s \in \mathcal{S}\) such that \(R_1(s) \geq R_2(s)\), the secrecy rate is determined by the achievable rate of wiretap channel and key rate. For sequential analysis, the number of message bits transmitted by stopping time \(\tau\) is  
\begin{align}
    M_\tau \geq  &\sum_{k=n^{\beta}+1}^{\tau/N} \mathbbm{1}\{\hat{S}_{(k-1)N}=s\} \times \nonumber\\ &\lfloor (1-\Delta)N (R_1(s)-R_2(s)+\min(R_2(s),R_{\text{key}}(s))) \rfloor.\nonumber
\end{align}
By Lemma 5 in Appendix\cite{guo2026secure} and sufficiently small \(\beta\)
\begin{align}
    \frac{M_\tau}{n} &\geq (1-\zeta_2) \frac{\tau}{n} \times \nonumber\\ &\lfloor (1-\Delta) (R_1(s)-R_2(s)+\min(R_2(s),R_{\text{key}}(s))) \rfloor,\nonumber
\end{align}
for \(\zeta_2 >0\). Moreover, with standard concentration of measure argument we have 
\begin{align} \label{achievablerate}
    &\min_{m\in \mathcal{M}} \mathbb{P}\bigg(R^n  \geq  (1-\zeta_3)(1-\zeta_2)(1-\Delta) \\
    &\times(R_1(s)-R_2(s)+\min(R_2(s),R_{\text{key}}(s))\bigg)\geq 1- e^{-n\zeta_4}.\nonumber
\end{align}
Similarly, for \(R_1(s) < (R_2(s)-R_{\text{key}}(s))^+\), we have 
\begin{equation}\label{achievablerate2}
    \mathbb{P}(R^n \geq \min(R_1(s),R_{\text{key}}(s)))\geq 1- e^{-n\zeta_4}.
\end{equation}
Finally, for \((R_2(s)-R_{\text{key}}(s))^+<R_1(s) < R_2(s)\), the probabilistic achievable rate is simply a time-sharing \(\rho\) between (\ref{achievablerate}) and (\ref{achievablerate2}). 

\section*{Acknowledgment}{Parts of this document have received assistance from generative AI tools to aid in the composition; the authors have reviewed and edited the content as needed and take full responsibility for it. This work was supported in part by the National Science Foundation (NSF) under grant 2148400 as part of the Resilient \& Intelligent NextG Systems (RINGS) program.}

\clearpage            
\balance            

\bibliographystyle{IEEEtran}
\bibliography{references}

\newpage
\appendices
\onecolumn
\section*{Appendix A: Proof of Lemma \ref{stopping time lemma}}
For any \(s,s'\in \mathcal{S}, s\neq s'\), we define the pair-wise log-likelihood ratio (LLR)
\begin{equation} \label{pairwiseLLR}
    L^{(1)}_{s,s';t}=\log \frac{\mathbb{P}_s(x_t,y_{1,t})}{\mathbb{P}_{s'}(x_t,y_{1,t})}.
\end{equation}
The probability that stopping time \(\tau\) exceeds \(n\) is given by 

\begin{align} \label{following}
    \mathbb{P}_s(\tau >n) & \leq \sum_{s'\neq s}\mathbb{P}\left(\log \frac{\mathbb{P}_{s}(x^{n-N},y_1^{n-N})}{\mathbb{P}_{s'}(x^{n-N},y_1^{n-N})}< n(\mathbb{D}(W_{Y_1|X,s}||W_{Y_1|X,s'}|P_{X,s})-\epsilon)\right) \nonumber \\
    & =\sum_{s'\neq s}\mathbb{P} \bigg(\sum_{t=1}^{n-N}\bigg[L^{(1)}_{s,s';t}-\mathbb{E}\bigg[L^{(1)}_{s,s';t}|\mathcal{F}_{t-1}\bigg]\bigg] < \nonumber \\
     & \quad \quad \quad \quad \quad n(\mathbb{D}(W_{Y_1|X,s}||W_{Y_1|X,s'}|P_{X,s})-\epsilon)-\sum_{t=1}^{n-N}\mathbb{E}\left[L^{(1)}_{s,s';t}|\mathcal{F}_{t-1}\right]\bigg),
\end{align}

where \(V_{s,s';t'} = \sum_{t=1}^{t'}L^{(1)}_{s,s';t}-\mathbb{E}\left[L^{(1)}_{s,s';t}|\mathcal{F}_{t-1}\right]\) is a martingale adapted to filtration \(\mathcal{F}_{t'}^\infty \triangleq \sigma (X_{t'},Y_{1,t'},M)\). Define the maximum difference 
\begin{equation}
    \phi_{s,s'} = \underset{x,y}{\max}\,\left|\log \frac{W_{Y_1|X,s}(y|x)}{W_{Y_1|X,s'}(y|x)}-\mathbb{D}(W_{Y_1|X,s}||W_{Y_1|X,s'}|P_{X,s})\right|
\end{equation}
for which 
\begin{equation}
    \left|L^{(1)}_{s,s';t}-\mathbb{E}\left[L^{(1)}_{s,s';t}|\mathcal{F}_{t-1}\right]\right| < \phi_{s,s'}
\end{equation}
holds almost surely. Observe that \(\phi_{s,s'}\) is finite Bby Assumption 2). Then, from Lemma \ref{stopping concentration} below and construction of the policy, \(\mathbb{E}\left[L^{(1)}_{s,s';t}|\mathcal{F}_{t-1}\right] = \mathbb{D}(W_{Y_1|X,s}||W_{Y_1|X,s'}|P_{X,s})\) for all \(t >N n^{\beta}\) such that \(t-(\lfloor \frac{t-1}{N}\rfloor)N \leq (1-\Delta)N\), then for sufficiently small \(\beta\)
\begin{align}
    \frac{1}{n}\sum_{t=1}^{n-N} \mathbb{E}\left[L^{(1)}_{s,s';t}|\mathcal{F}_{t-1}\right] & = \frac{1}{n}\left[\sum_{t=1}^T \mathbb{E}\left[L^{(1)}_{s,s';t}|\mathcal{F}_{t-1}\right] + \sum_{t=T+1}^{n-N}\mathbb{E}\left[L^{(1)}_{s,s';t}|\mathcal{F}_{t-1}\right]\right] \nonumber\\
    & \geq\frac{1}{n}\left[\sum_{t=1}^T \mathbb{E}\left[L^{(1)}_{s,s';t}|\mathcal{F}_{t-1}\right] + (1-\Delta)(n-N-T)\mathbb{D}(W_{Y_1|X,s}||W_{Y_1|X,s'}|P_{X,s})\right].
\end{align}
By Lemma \ref{stopping concentration} \(T=o(n)\) and \(N =o(n)\) by definition, we conclude that 
\begin{equation}
    \frac{1}{n}\sum_{t=1}^{n-N} \mathbb{E}\left[L^{(1)}_{s,s';t}|\mathcal{F}_{t-1}\right] \geq (1-\epsilon_3)(1-\Delta)\mathbb{D}(W_{Y_1|X,s}||W_{Y_1|X,s'}|P_{X,s}), 
\end{equation}
for any \(\epsilon_3 >0\). Therefore, 
\begin{equation}
   \psi_{s,s'} = \frac{1}{n}\sum_{t=1}^{n-N} \mathbb{E}\left[L^{(1)}_{s,s';t}|\mathcal{F}_{t-1}\right]-\mathbb{D}(W_{Y_1|X,s}||W_{Y_1|X,s'}|P_{X,s}) +\epsilon
\end{equation}
can be made positive for \(\epsilon>0\), by choosing sufficiently small \(\Delta\) and \(\epsilon_3\). Finally, by applying Azuma's inequality, for \(n\) sufficiently large, we obtain 
\begin{equation}
    \mathbb{P}_s(\tau >n) \leq (|\mathcal{S}|-1) \exp \left(- \min_{s'\neq s}\frac{n \psi_{s,s'}^2}{2\phi^2_{s,s'}}\right).
\end{equation}

\begin{lemma}\label{stopping concentration}
    Define \(T\) to be the earliest time at which \(\hat{s}_T = s\), where \(s\) is the true state, then there exists \(c>0\) and \(c_0>0\) such that 
\begin{equation}
    \mathbb{P}(T>t) \leq c e^{-c_0 t}.
\end{equation}
\end{lemma}
\begin{proof}
This is a well known result by simply applying Chernoff's inequality, which we provide for completion. Based on the log-likelihood ratio, \(\psi_1(y_1^t) \neq s\) whenever

\begin{equation}
   \exists s'\neq s, \log \frac{\mathbb{P}_{s}(y_1^t,x^t)}{\mathbb{P}_{s'}(y_1^t,x^t)} \leq 0.
\end{equation}
Then, 

\begin{align}
    \mathbb{P}_s(T>t) & \leq \mathbb{P}\left(\exists s'\neq s, \log \frac{\mathbb{P}_{s}(y_1^t,x^t)}{\mathbb{P}_{s'}(y_1^t,x^t)} \leq 0\right) \nonumber\\
    & \leq \sum_{s'\neq s} \mathbb{P}\left( \log \frac{\mathbb{P}_{s}(y_1^t,x^t)}{\mathbb{P}_{s'}(y_1^t,x^t)} \leq 0\right) \nonumber\\
    & = \sum_{s'\neq s} \mathbb{P}\left( e^{\lambda\sum_{k=1}^t \log \frac{W_{Y_1|X,s}(y_{1,k}|x_k)}{W_{Y_1|X,s'}(y_{1,k}|x_k)}} \geq 1\right) \nonumber\\
    & \overset{(a)}{\leq} \sum_{s'\neq s} \mathbb{E}_{s}\left( e^{\lambda\sum_{k=1}^t \log \frac{W_{Y_1|X,s}(y_{1,k}|x_k)}{W_{Y_1|X,s'}(y_{1,k}|x_k)}} \right) \nonumber\\
    & \leq (|\mathcal{S}|-1)e^{t \max_{s}\max_{s'\neq s} \upsilon_{s,s'}(\lambda)},
\end{align}
where (a) follows by applying Chernoff's bound. 
For any \(\lambda<0\), 
\begin{align}
    \upsilon_{s,s'}(\lambda)& = \log \mathbb{E}_{s}\left[e^{\lambda\log\frac{W_{Y_1|X,s}(y_{1}|x)}{W_{Y_1|X,s'}(y_{1}|x)}}\right] \nonumber\\
    & = \log \sum_{x,y}P_{X,s}(x) W_{Y_1|X,s}(y_1|x)\left(\frac{W_{Y_1|X,s}(y_1|x)}{W_{Y_1|X,s'}(y_1|x)}\right)^{\lambda}
\end{align}
is the state-dependent moment generating function. Taking \(c = |\mathcal{S}|-1\) and \(c_0 = -\max_{\lambda<0}\max_{s} \max_{s'\neq s}\upsilon_{s,s'}(\lambda) >0\) concludes the proof. 
\end{proof}

\section*{Appendix B: Proof of \ref{main exponent}}

The core idea for proving Tx's achievable exponent is based on the fact that the standard measure concentration argument \cite{nitinawarat2013controlled} applies on a packet-level adaptation. For \(n\) sufficiently large, the probability of detection error can be bounded as 
\begin{equation} \label{Pd1}
    P_{d_1} \leq \underset{s \in \mathcal{S}}{\text{max}}\,\sum_{s'\neq s}\mathbb{P}_s\left(\sum_{t=1}^\tau L^{(1)}_{s,s';t} \leq n (\mathbb{D}(W_{Y_1|X,s}||W_{Y_1|X,s'}|P_{X,s})-\epsilon)\right). 
\end{equation}
But the pairwise LLR can be expressed as 
 \begin{align}
     \frac{1}{n} \sum_{t=1}^\tau L^{(1)}_{s,s';t} 
      =  &\frac{1}{n} \sum_{t=1}^\tau \left\{ L^{(1)}_{s,s';t} - \mathbb{E}[L^{(1)}_{s,s';t} |\mathcal{F}_{t-1}] \right\} \nonumber\\
     +& \frac{1}{n} \sum_{t=1}^\tau \left \{\mathbb{E}[L^{(1)}_{s,s';t} |\mathcal{F}_{t-1}] - \mathbb{D}(W_{Y_1|X,s}||W_{Y_1|X,s'}|P_{X,s})\right\} \nonumber \\
     + &\frac{\tau}{n}\mathbb{D}(W_{Y_1|X,s}||W_{Y_1|X,s'}|P_{X,s}).
    \end{align}
The two terms in \(\{\}\) concentrate to zero exponentially fast \cite{chernoff1992sequential} \cite{nitinawarat2013controlled}, which implies \(\mathbb{E}[\frac{1}{n} \sum_{t=1}^\tau L^{(1)}_{s,s';t}] \rightarrow \frac{\tau}{n}\mathbb{D}(W_{Y_1|X,s}||W_{Y_1|X,s'}|P_{X,s})\) for sufficiently large \(n\). Applying Chernoff's bound to (\ref{Pd1}) and Lemma 3 gives 
\begin{align}
    P_{d_1} \leq &\underset{s \in \mathcal{S}}{\text{max}}\,\sum_{s'\neq s} \exp(-n(\mathbb{D}(W_{Y_1|X,s}||W_{Y_1|X,s'}|P_{X,s})-\epsilon)) \nonumber\\
    \leq &\underset{s \in \mathcal{S}}{\text{max}}\,(|\mathcal{S}|-1)\exp(-n (\min_{s'\neq s}\,\mathbb{D}(W_{Y_1|X,s}||W_{Y_1|X,s'}|P_{X,s})-\epsilon)).
\end{align}

\section*{Appendix C: Proof of Lemma \ref{Eve exponent lemma}}

We first focus on the resolvability region in which \(R_1(s) +R_{\text{key}}(s)> R_2(s)\). 
Define the LLR at Eve as
\begin{equation}
    L^{(2)}_{s,s';t}=\log \frac{\mathbb{P}_s(y_{2,t})}{\mathbb{P}_{s'}(y_{2,t})}.
\end{equation}
Also define the set \(\mathcal{T}_s\triangleq \{t \in [1,n]: \{t\geq T\} \cap \{t-(\lfloor \frac{t-1}{N}\rfloor)N \leq (1-\Delta)N\}\}\). Then Eve's detection error probability between \(s\) and \(s'\) is given by the optimal LLR test as
\begin{align}
    &\mathbb{P}_s(f^{s,s'}_{ML}(Y_2^{\tau}) \neq s |\tau = n) \nonumber\\
    = & \mathbb{P}_s\left(\sum_{t=1}^{n} L^{(2)}_{s,s';t} \leq 0 \right) \nonumber\\
    = & \mathbb{P}_s\left(\sum_{t\in \mathcal{T}_s} L^{(2)}_{s,s';t} + \sum_{t\notin \mathcal{T}_s} L^{(2)}_{s,s';t}\leq 0 \right) \nonumber\\
    = & \mathbb{P}_s\left(\frac{1}{|\mathcal{T}_s|}\sum_{t\in \mathcal{T}_s} L^{(2)}_{s,s';t} \leq  \frac{1}{|\mathcal{T}_s|}\sum_{t\notin \mathcal{T}_s} L^{(2)}_{s',s;t} \right).
\end{align}
Since the messages are protected with keys, Eve's hypothesis testing is based on the mixture 
\begin{align}  \label{nonchernoff}
&\frac{-1}{n}\log\mathbb{P}_s(f^{s,s'}_{ML}(Y_2^{\tau}) \neq s |\tau = n) \nonumber\\
    \leq& \frac{-1}{n}\log\mathbb{P}_s\left(\frac{1}{|\mathcal{T}_s|}\sum_{t\in \mathcal{T}_s} L^{(2)}_{s,s';t} \leq  0 \right) \nonumber \\
    = & \frac{-1}{n}\log\mathbb{P}_s\left(\frac{1}{|\mathcal{T}_s|}\sum_{t\in \mathcal{T}_s} L^{(2)}_{s',s;t} \geq  0 \right) \nonumber \\
    \overset{(a)}{\leq} &   \frac{-1}{n}\log \mathbb{E} \left[e^{\lambda \frac{1}{|\mathcal{T}_s|} \sum_{t\in \mathcal{T}_s} L^{(2)}_{s',s;t}}\right] \nonumber \\
    = &  \sup_{\lambda >0}\, \frac{-1}{n}\log \mathbb{E}\left[\exp\left( \log \prod_{t\in \mathcal{T}_s}\frac{\mathbb{P}_{s'}(y_{2,t})}{\mathbb{P}_{s}(y_{2,t})}\right)^{\frac{\lambda}{|\mathcal{T}_s|}}\right] \nonumber\\
    =&  \sup_{\lambda >0}\, \frac{-1}{n}\Bigg(\frac{1}{|\mathcal{T}_s|}\sum_{t\in \mathcal{T}_s}\log\mathbb{E}_{y_2 \sim  \frac{1}{2^{M_\tau}}\sum_{m=1}^{2^{M_\tau}} W_{Y_{2}|X,s'}(y_{2,t}|x_t(m))}\Bigg[\frac{\sum_{m=1}^{2^{M_\tau}} W_{Y_{2}|X,s'}(y_{2,t}|x_t(m))}{ \sum_{m=1}^{2^{M_\tau}} W_{Y_{2}|X,s}(y_{2,t}|x_t(m))}^\lambda \Bigg]\Bigg) \nonumber\\
    \overset{(b)}{=}& \sup_{\lambda \in [0,1]}\, \frac{-1}{n} \log \bigg( (1-\Delta)(1-\frac{n - |\mathcal{T}_s|}{n})\sum_{y_2^n} \bigg(P_{Y^n|\mathcal{C}^n_{2^{M_n}}(P_{X,s})}(y_2^n)\bigg)^\lambda\bigg(P_{Y^n|\mathcal{C}^n_{2^{M_n}}(P_{X,s'})}(y_2^n)\bigg)^{1-\lambda}\bigg)\nonumber\\
    \overset{(c)}{=}& \frac{-1}{n}\log \mathbb{C}^{M_n}+\delta_6,
\end{align}
for \(\delta_6 >0\), where \(\mathbb{C}^{M_n}=2^{-n \mathbb{C}(P_{Y^n|\mathcal{C}^n_{2^{M_n}}(P_{X,s})}(y_2^n)||P_{Y^n|\mathcal{C}^n_{2^{M_n}}(P_{X,s'})}(y_2^n))}\) is the non-log Chernoff information. \((a)\) comes from Chernoff inequality and \((b)\) is supported by the fact that for any \(s \in \mathcal{S}, t\in \mathcal{T}_s\), codewords \(x(m)\) are drawn i.i.d from the distribution \(P_{X,s}\). 
\((c)\) is based on \(|\mathcal{T}_s| = \Theta(n), \Delta >0\) and Lemma \ref{stopping time lemma}.
Moreover, when \(R_1(s) +R_{\text{key}}(s)> R_2(s)\), by soft covering results
\begin{equation}
    \frac{1}{2}\Bigg|\Bigg|\frac{1}{2^{M_n}} \sum_{m=1}^{2^{M_n}} W_{Y_{2}|X,s}^{\otimes n}(y^n_2|x^n(m))-\sum_{x^n} P_{X,s}(x^n) W_{Y_2|X,s}(y^n_2|x^n)\Bigg|\Bigg|_1 \leq e^{-n E_{\textnormal{SC}}},
\end{equation}
for \(E_{\textnormal{SC}}>0\) given in Lemma \ref{softcovering}. Then, asymptotically we have
\begin{align}
    \mathbb{C}^{M_n} &\geq 1-\mathbb{V}\bigg(\frac{1}{2^{M_n}} \sum_{m=1}^{2^{M_n}} W_{Y_{2}|X,s}^{\otimes n}(y^n_2|x^n(m)),\frac{1}{2^{M_n}} \sum_{m=1}^{2^{M_n}} W_{Y_{2}|X,s'}^{\otimes n}(y^n_2|x^n(m))\bigg) \\
    & \geq 1-\frac{1}{2}\Bigg|\Bigg|\frac{1}{2^{M_n}} \sum_{m=1}^{2^{M_n}} W_{Y_{2}|X,s}^{\otimes n}(y^n_2|x^n(m))-\sum_{x^n} P_{X,s}(x^n) W_{Y_2|X,s}(y^n_2|x^n)\Bigg|\Bigg|_1 \nonumber\\
     & \quad\quad-\frac{1}{2}\Bigg|\Bigg|\frac{1}{2^{M_n}} \sum_{m=1}^{2^{M_n}} W_{Y_{2}|X,s}^{\otimes n}(y^n_2|x^n(m))-\sum_{x^n} P_{X,s'}(x^n) W_{Y_2|X,s'}(y^n_2|x^n)\Bigg|\Bigg|_1\nonumber\\
    & \quad\quad- \frac{1}{2}\Bigg|\Bigg|\sum_{x^n}P_{X,s}(x^n) W_{Y_2|X,s}(y^n_2|x^n)-\sum_{x^n} P_{X,s'}(x^n) W_{Y_2|X,s'}(y^n_2|x^n)\Bigg|\Bigg|_1\nonumber\\
    & \geq (1-\mu)e^{-n\mathbb{C}(P_{X,s}\circ W_{Y|X,s}||P_{X,s'}\circ W_{Y|X,s'})}-2e^{-\min (E_{\textnormal{SC}}(P_{X,s},W_{Y|X,s}),E_{\textnormal{SC}}(P_{X,s'},W_{Y|X,s'})n},
\end{align}
for \(\mu >0\), where \(E_{\textnormal{SC}} \) is the soft covering exponent given in lemma \ref{softcovering}. Finally, we have for arbitrarily small \(\mu\)
\begin{align}
    \underset{n\rightarrow \infty}{\lim} -\frac{1}{n} \log P_{d_2} &\leq \mathbb{C}(P_{X,s}\circ W_{Y|X,s}||P_{X,s'}\circ W_{Y|X,s'})\nonumber \\&-\frac{1}{n}\log\bigg(1-c_3 \frac{e^{-n\min (E_{\textnormal{SC}}(P_{X,s},W_{Y|X,s}),E_{\textnormal{SC}}(P_{X,s'},W_{Y|X,s'}))}}{e^{-n \mathbb{C}(P_{X,s}\circ W_{Y|X,s}||P_{X,s'}\circ W_{Y|X,s'})}}\bigg).
\end{align}
We next turn our attention to the situation in which \(R_1(s) < (R_2(s)-R_{\text{key}}(s))^+\). The difficulty in analyzing this regime is that there exists no tight upper bound on the Chernoff information between mixtures in (\ref{nonchernoff}) while preserving the dependency on codebook size. Since the distribution at Eve is no longer resolvable, the additional structure induced by the codebook mixtures may enable Eve to infer the type of the codebook \(P_{X,s}\). In this case, Eve may recover the codeword types and obtain state information from both public policy and observation. Therefore, Tx uses \(P_{X,s} = P_X,  \forall s\in \mathcal{S}\) and retains the advantage via control of the stopping time. This setting has been analyzed in \cite{chang2021evasive}, and we adapt to our block-based adaptation to show that Eve's detection exponent is (details are omitted for brevity).

\begin{equation}
    \underset{n\rightarrow \infty}{\lim} -\frac{1}{n} \log P_{d_2} \leq \mathbb{C}(W_{Y|X,s}|| W_{Y|X,s'}|P_{X,s})
\end{equation}

\begin{lemma}\label{softcovering}
Let \(P_X\) be a fixed \(n\)-type and \(\mathcal{C}^n_M(P_X)\) be a constant composition codebook. For any non-degenerate discrete memoryless channel \(W_{Y^n|X^n}: \mathcal{X}^n \rightarrow \mathcal{Y}^n\) and rate \(R \geq \mathbb{I}(P_X,W_{Y|X})\), let \(M = \lceil \exp (nR)\rceil\), then 
\begin{equation}
   \bigg|\bigg|P_{Y^n|\mathcal{C}^n_M(P_X)}(y^n)-P^n_{X}\circ W^n_{Y|X}\bigg|\bigg|_1  \leq e^{-n E_{\textnormal{SC}}(P_X,W_{Y|X})},
\end{equation}
where 
\begin{equation}
    E_{\textnormal{SC}}(P_X,W_{Y|X}) = \min_{Q_{Y|X} \in W_{Y|X}} \bigg\{\mathbb{D}(P_X Q_{Y|X} ||P_X W_{Y|X})+\frac{1}{2}[R - \mathbb{D}(P_X Q_{Y|X}||P_X Q_Y)]^+\bigg\},
\end{equation}
and \(Q_Y = P_X \circ Q_{Y|X}\) is the marginal on \(Y\), for \(n \rightarrow \infty\).
\end{lemma}

\section*{Appendix D:  Reliability Analysis}

Using the union bound, the decoding error probability is given by 
\begin{align} \label{reliability}
    P_c \leq &\sum_{k=n^{\beta}+1}^{\tau/N} \mathbb{P}(g_1(Y_{1,(k-\Delta)N+1:kN})\neq \hat{S}_{(k-1)N}|\hat{S}_{(k-1)N} = s) \nonumber \\
    + & \sum_{k=n^{\beta}+1}^{\tau/N}\mathbb{P}(g_1(Y_{1,(k-1)N+1:(k-\Delta)N})\neq M[M_{(k-1)N};M_{kN}]|M=m,\hat{S}_{(k-1)N}=s) \nonumber\\
    + & \sum_{k=n^{\beta}+1}^{\tau/N} \mathbb{P}(\psi_1(Y_{1}^{(k-1)N},X^{(k-1)N})\neq S|S=s). 
\end{align}
By Lemma~\ref{stopping concentration}, the third term in (\ref{reliability}) vanishes. The first two terms in (\ref{reliability}) are the decoding error probability of the estimated state and message, respectively. For \(R_1(s)\geq R_2(s)\), decoding reliability is guaranteed by standard wiretap code \cite{bloch2011physical}. For \(R_1(s)< R_2(s)\), vanishing decoding error probability comes from crypto lemma \cite{bloch2011physical} and constant composition code results \cite{csiszar2011information}.

\end{document}